\documentclass[preprint]{article}

\usepackage[preprint,nonatbib]{neurips_2021}

\usepackage[utf8]{inputenc} 
\usepackage[T1]{fontenc}    
\usepackage{hyperref}       
\usepackage{url}            
\usepackage{booktabs}       
\usepackage{amsfonts}       
\usepackage{nicefrac}       
\usepackage{microtype}      
\usepackage{xcolor}         

\usepackage{textcomp}                  
\usepackage{mathptmx}                  
\usepackage{tabu}                      
\usepackage{booktabs}    
\usepackage{amsmath,amssymb,amsthm,float}
\usepackage[ruled,lined,commentsnumbered]{algorithm2e}
\usepackage{enumitem}
\usepackage{hyperref}
\usepackage{graphicx}
\graphicspath{{figures/}{pictures/}{images/}{./}} 

\newcommand{\RR}{\mathbb{R}}
\newcommand{\set}[1]{\left\{#1\right\}}
\newcommand{\bigp}[1]{\left(#1\right)}
\newcommand{\argmax}{\arg\max}
\newcommand{\argmin}{\arg\min}
\newcommand{\rvline}{\hspace*{-\arraycolsep}\vline\hspace*{-\arraycolsep}}
\newcommand{\one}[1]{\mathbf{1}_{#1}}
\newcommand{\tr}{\mathrm{tr}}
\newcommand{\rank}{\mathrm{rank}}

\newtheorem{theorem}{Theorem}
\newtheorem{corollary}[theorem]{Corollary}

\title{K-ARMA Models for Clustering Time Series Data}

\author{%
  Derek O. Hoare \\
  Department of Statistics and Data Science\\
  Cornell University\\
  \texttt{doh24@cornell.edu} \\
  \And
  David S. Matteson \\
  Department of Statistics and Data Science \\
  Cornell University \\
  \texttt{dm484@cornell.edu} \\
  \AND
  Martin T. Wells \\
  Department of Statistics and Data Science \\
  Cornell University \\
  \texttt{mtw1@cornell.edu}\\
}

\usepackage{setspace}

\begin{document}

\allowdisplaybreaks

\maketitle

\begin{abstract}
    We present an approach to clustering time series data using a model-based generalization of the K-Means algorithm which we call K-Models. We prove the convergence of this general algorithm and relate it to the hard-EM algorithm for mixture modeling. We then apply our method first with an AR($p$) clustering example and show how the clustering algorithm can be made robust to outliers using a least-absolute deviations criteria. We then build our clustering algorithm up for ARMA($p,q$) models and extend this to ARIMA($p,d,q$) models. We develop a goodness of fit statistic for the models fitted to clusters based on the Ljung-Box statistic. We perform experiments with simulated data to show how the algorithm can be used for outlier detection, detecting distributional drift, and discuss the impact of initialization method on empty clusters. We also perform experiments on real data which show that our method is competitive with other existing methods for similar time series clustering tasks.

\end{abstract}

\section{Introduction}

Today, an abundance of information is collected over time. This ranges from climate data to server metrics, biometric to sensor data, and digital signals to logistics information. With technological advances, we are also capable of collecting more time series data than ever before. Large quantities of time series data, however, come with the inability for an individual to inspect all of the data. As a consequence, we need methods for identifying patterns in large quantities of time series data that require only minimal manual overhead. One of the ways we can do this is by clustering the time series data and performing analysis at the cluster level. 



Several realizations of time series from the same generative process may look very different from each other yet still have the same underlying process. In these cases, a model based clustering approach can do a much better job of clustering time series that were generated from the same process than traditional point-based clustering methods like the K-Means algorithm. Model-based time series clustering methods have proven to be useful in many real-world situations. They have been used to cluster gene expression data \cite{venkatraman2021empirical}, to understand within-person dynamics \cite{bulteel2016}, to better understand currency exchange dynamics \cite{ROUT20147}, to cluster input signals \cite{elamin2020clustering}, and more. We study the time series clustering problem using new model-based clustering methods focusing on the general class of Autoregressive Moving Average models.


Autoregressive Moving Average (ARMA) processes and their Integrated extension (ARIMA) can be used to describe many observed time series, so we may know or strongly suspect that a collection observed time series come from an ARMA process. In such a case, we should take advantage of this underlying ARMA structure when clustering time series. This has led to several methods for clustering ARMA time series including approaches using features derived from the time series, mixture modeling approaches, methods which cluster the raw time series data \cite{baragona2001simulation}, and models fit to entire clusters of data \cite{bulteel2016,takano2020,hautamaki2008time}. Surveys of time series clustering methods can be found in \cite{decadereview,maharaj2019time,rani2012}. 

The method we explore in this paper uses a model-based generalization of the K-Means algorithm which we call K-Models. We apply this general framework to clustering ARMA time series by fitting ARMA models to entire clusters of data during the clustering process. The approach we use is most similar in flavor to the ARMA mixture modeling approach \cite{armaEM}. However, where they use an expectation maximization (EM) algorithm to fit mixtures of ARMA models to their time series and identify clusters based on membership probabilities, the K-Models approach takes a hard thresholding approach with potential convergence advantages.

\subsection{A Motivating Example}




\begin{figure}[b!]
    \centering
    \makebox[\textwidth][c]{
    \includegraphics[width = 1.2\linewidth]{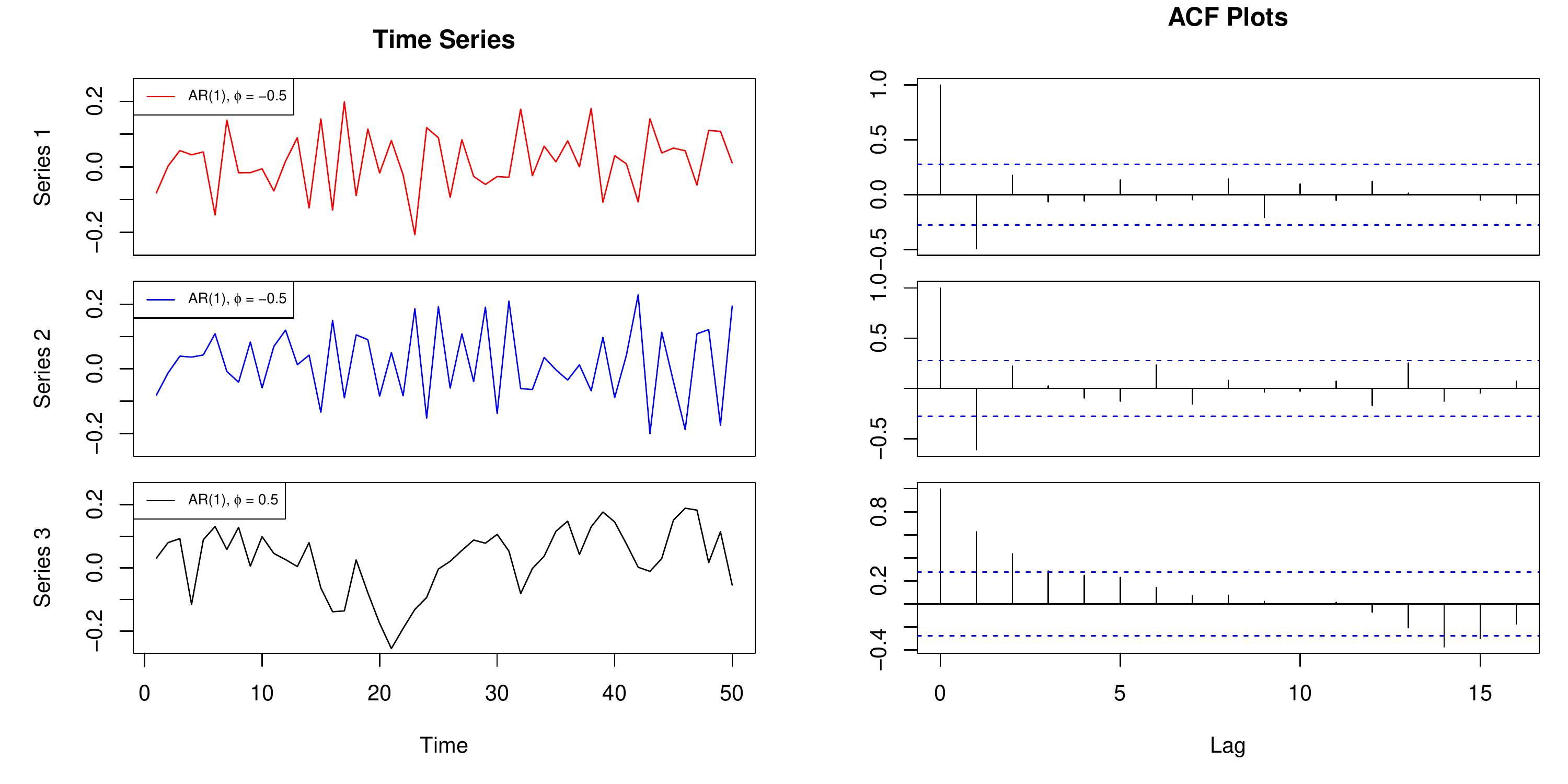}
    }
    \caption{Three AR(1) time series. The Euclidean distances between the time series are d(Series 1, Series 3) = 1.049016, d(Series 2, Series 3) = 1.215401, d(Series 1, Series 2) = 1.462759. Looking at the autocorrelation plots we see that Series 1 and Series 2 appear more similar to each other than to Series 3.}
    \label{fig:motivation}
\end{figure}

Our main reason for studying model-based clustering is that it is able to capture similarities that distance-based clustering methods cannot. As a motivating example, consider the three time series in Figure~\ref{fig:motivation}. Series 1 and Series 2, shown in red and blue were both generated from the exact same AR(1) process with coefficient $\phi_1 = -0.5$. Series 3, shown in black, was generated from an AR(1) process with coefficient $\phi_1 = 0.5$. When we look at the sums of squared differences between the AR(1) time series, however, we would conclude that the Series 1 (red) is more similar to Series 3 (black) than it is to Series 2 (blue), even though Series 1 and Series 2 were generated from the same AR(1) process. This suggests that traditional clustering methods that rely on distance measure between the time series such as the Euclidean distance are unable to capture similarities in the generative processes of time series. On the other hand, if we look at the fitted AR(1) coefficients we have $\hat\phi_1^{red} = -0.4403$, $\hat\phi_1^{blue} = -0.6039$, and $\hat\phi_1^{black} = 0.6449$. Looking at these model coefficients, Series 1 and Series 2 appear much more similar. Clustering these series using the model-based approach that we develop in this paper using 2 clusters also identifies Series 1 and Series 2 as belonging to one cluster and the black as belonging to another.

Being able to identify time series that come from the same process is important in the applications we have discussed. For clustering input signals, the clusters represent that the inputs come from the same process, so the underlying process is more important than the specific realizations of that process. For clustering gene expressions in response to a treatment, it is the trend in how a gene is regulated in response to that treatment that is important, not the magnitude of the response.



\subsection{Organization}
In Section~\ref{sec:kmeans} we describe a general algorithm for model-based clustering which we call K-Models, and prove a condition for the convergence of the K-Models algorithm. In Section~\ref{sec:ar} use this K-Models framework to cluster time series using AR($p$) models using least squares and least absolute deviations approaches.
In Section~\ref{sec:arma} we explore the use of the K-Models framework for clustering ARMA/ARIMA time series. In doing so, we study methods for fitting a single ARMA model to a collection of time series which are used in our clustering algorithm. In Section~\ref{sec:diagnostic} we develop and use an extension of the Ljung-Box test to answer questions about the goodness of fit of ARMA models to our clusters to detect outliers and distributional drift in our time series data. In Section~\ref{sec:vanishing}, we study different initialization methods for our algorithm. In doing so, we find that certain initialization methods lead to many empty clusters in cases where the number of clusters we fit is greater than the true number of underlying clusters. In Section~\ref{sec:realdata}, we test our algorithm on a few open source datasets for comparison to other clustering methods. Finally, in Section~\ref{sec:discussion} we discuss limitations and possible extensions of our clustering approach.

\section{K-Means and K-Models}\label{sec:kmeans}

The K-Means algorithm is one of the most widely used unsupervised clustering algorithms. For a given set of data points $X_1, \ldots, X_n\in\mathbb{R}^d$, the K-Means clustering algorithm seeks to minimize the within-groups sums of squared error over all possible cluster centers $\left\{m_1, \ldots, m_k\right\}\subseteq \mathbb{R}^d$ by using an iterative algorithm of assigning points to clusters and then updating the cluster centers \cite{kmeans, macqueen1967some}. If we view the cluster centers $m_i$ as a constant function representation of a cluster, then we can view K-Means as a model-based clustering algorithm where we assign each point to the constant function that best represents that point.

This observation about the K-Means algorithm has led to several derivative model-based clustering algorithms including Clusterwise Linear Regression where the models used are linear models \cite{spath1979}, K-Random Forests where the models used are Random Forests \cite{bicego2019}, K-Means based kernel clustering algorithms \cite{camastra}, and has been described for some general dissimilarity measures \cite{maranzana1963, bock2008}. We call the model-based generalization of the K-Means algorithm \textit{K-Models} (see Algorithm~\ref{alg:kmodels}). 


\begin{algorithm}[t]
    \caption{The K-Models Clustering Algorithm}\label{alg:kmodels}
    \DontPrintSemicolon
    \KwData{$X_1, \ldots, X_n$, a model class $\mathcal{M}$, a loss function $L$, and a number of clusters $k$.}
    \KwResult{A set of clusters $C_1, \ldots, C_k$ and a set of models $M_1, \ldots, M_k \in \mathcal{M}$}
    \hrulefill\\
    
    Randomly initialize $M_i\in\mathcal{M}$ for $i=1, \ldots, k$ with $M_i\neq M_j$ for $i\neq j$.\;
    \While{not converged}{
    \For{i=1 \emph{\KwTo} k}{
    $C_i \gets \left\{X_l : L(X_l, M_i)\leq L(X_l, M_j) \forall j\neq i\right\}$ \tcp*{Assignment Step}
    } 
    \For{i=1 \emph{\KwTo} k}{$\displaystyle M_i \gets \hat{M}(C_i)$     \tcp*{Update Step}
    }
    }
\end{algorithm}

We point out that the update step in Algorithm~\ref{alg:kmodels}, typically the most computationally intensive step, can be computed in parallel across the clusters of data. Furthermore, we are able to guarantee the convergence of the K-Models algorithm to a local optimum for a large class of models.

\begin{theorem}\label{thm:kmodels} Let $\mathcal{M}$ be a family of models. Let $X_1, \ldots, X_n\in\RR^d$. Let $L:\RR^d\times \mathcal{M}\to \RR_{\geq 0}$ be a loss function.
If $\hat{M}(C_i) = \arg\min_{M\in \mathcal{M}} \sum_{X\in C_i}L(X, M)$ then Algorithm~\ref{alg:kmodels} converges.
\end{theorem}

\begin{proof}
This proof follows the proof of convergence for the $K$-Means algorithm. At a given step in the algorithm, let $C = \set{C_i}$ and $M = \set{M_i}$ be our candidate clusters and their corresponding models. Consider the global loss function
\begin{equation*}
    E(C, M) = \sum_{i=1}^k \sum_{X\in C_i} L(X, M_i).
\end{equation*} At the assignment step we assign each point $X_i$ to the cluster in which they have the smallest loss to get a set of clusters $C' = \set{C_i'}$ which by construction satisfy the relationship
\begin{align*}
    E(C', M) = \sum_{i=1}^k \sum_{X\in C_i'} L(X, M_i)
    \leq \sum_{i=1}^k \sum_{X\in C_i} L(X, M_i) = E(C, M).
\end{align*}
For the update step, we construct $M' = \set{M_i'}$ so that $M_i'$ satisfies $\sum_{X\in C_i'}  L(X, M_i')\leq \sum_{X\in C_i'}  L(X, M_i).$ Thus 
\begin{align*}
    E(C', M') = \sum_{i=1}^k \sum_{X\in C_i'} L(X, M_i')
    \leq \sum_{i=1}^k \sum_{X\in C_i'} L(X, M_i) = E(C', M).
\end{align*}
So we have that the total loss $E$ is monotonically decreasing with each step of the iteration in the K-Models algorithm. Since $E$ is bounded below by $0$, this implies the convergence of $E$ to a local minimum. Therefore Algorithm~\ref{alg:kmodels} converges. \end{proof}

Theorem~\ref{thm:kmodels} says is that if we use the same loss function for both fitting the models and performing the cluster assignment, then Algorithm~\ref{alg:kmodels} converges to a local optimum solution. We can see that Theorem~\ref{thm:kmodels} implies the convergence of the K-Means algorithm since both the assignment and the update steps use the euclidean distance as the loss $L$. Additionally, we will use Theorem~\ref{thm:kmodels} to imply the convergence of our clustering algorithms by using a least squares criteria in the update step when we fit our time series models using least squares methods, and by using a least absolute deviation criteria in the update step when we fit models using median regression techniques.

\subsection{Relationship to Mixture Modeling and Expectation Maximization Algorithms}

The K-Models framework presents a hard clustering approach to model-based clustering which can accommodate many model classes. One of the competitors to this framework is the mixture modeling approach which has a nice probabilistic framework and a general fitting procedure, namely the EM algorithm \cite{dempster1977maximum}. This mixture modeling framework appears to be more commonly used even though it has many of the same flaws as the K-Models framework including convergence to a local optimum solution and convergence rates dependent on the initialization choice. 

As a part of the EM algorithm for mixture modeling, during the E-step class membership probabilities are computed. These are usually of the form 
\begin{equation*}
    P(X\in C_i) = \frac{f_{\theta_i}(X)}{\sum_{j=1}^kf_{\theta_j}(X)},
\end{equation*} where $f_{\theta_i}$ is the density of the $i$th mixture component with parameters $\theta_i$.  By taking the densities to a power $\alpha>0$, we get 
\begin{equation*}
P_\alpha(X\in C_i) = \frac{f_{\theta_i}(X)^\alpha}{\sum_{j=1}^kf_{\theta_j}(X)^\alpha}.
\end{equation*}
This $P_\alpha$ is often used for annealing the EM-algorithm by slowly increasing $\alpha$ from $0\to 1$ to help it get away from local solutions and reach a globally optimum solution \cite{ueda1998deterministic}.
Assuming no ties, taking a limit as $\alpha\to\infty$ yields 
\begin{equation*} 
P_\infty(X\in C_i) = \begin{cases}
1 & \text{if } i = \arg\max_j f_{\theta_j}(X);\\
0 & \text{otherwise.}
\end{cases}
\end{equation*}
By adding the parameter $\alpha$ and taking the limit as $\alpha\to\infty$, we have essentially converted the EM algorithm into a hard clustering algorithm since during the E-step we would be assigning elements to the class in which they have the highest likelihood, and during the M-step we would compute the parameters based on those class assignments. Using $P_\infty$ for the cluster membership probabilities is sometimes called hard-EM. The E-step is analogous to the assignment step in Algorithm~\ref{alg:kmodels} and the M-step is analogous to the update step.

We point this out because increasing $\alpha$ to be greater than $1$ has been used to hasten the convergence rate of the EM algorithm \cite{naim2012convergence}, and K-Means algorithms have in several cases been shown to have better convergence rates than EM algorithms \cite{bottou1994convergence, su2007search}. This also helps explain why using K-Means can be useful for initializing EM runs for Gaussian mixture models, and how we can extend this idea to using K-Models to initialize EM Algorithms for more complicated mixture modeling problems. 

The relationship between K-Models and EM algorithms for mixture modeling can be summarized with the following theorem.

\begin{theorem}\label{thm:hardem}
Hard-EM for mixture modeling is a special case of K-Models if the likelihood function is bounded above.
\end{theorem}
\begin{proof} Consider data $X = \set{X_1, \ldots, X_n}$ and $k$ mixture components of the densities $f_\theta$.
If the likelihood function is bounded above, then $loglik(\theta, X)<C$ for all $\theta$, so we can consider $M(\theta, X) = C - loglik(\theta, X) >0$. Finding the value of $\theta$ that minimizes $M(\theta, X)$ is equivalent to finding the value of $\theta$ that maximizes $lik(\theta, X)$, and $M(\theta, X)$ satisfies the properties of a loss function for Algorithm~\ref{alg:kmodels}. 

\textit{Equivalence of hard E-step and K-Models assignment Step}: In the hard-EM algorithm, we have $P_\infty(X\in C_j)=1$ if and only if $f_{\theta_\ell}(X)<f_{\theta_j}(X)$ for all $\ell\neq j$. This happens if and only if $C-\log f_{\theta_\ell}(X)>C-\log f_{\theta_j}(X)$ for all $\ell\neq j$. The likelihood $lik(\theta, X) = f_{\theta}(X),$ so we have $M(\theta_\ell, X)> M(\theta_j, X)$ for all $\ell\neq j$ and thus we have that $P_\infty(X\in C_j) = 1$ is equivalent to the assignment rule $X\in C_j$ if and only if $M(\theta_j, X)<M(\theta_\ell, X)$ which is the assignment rule we have in the K-Models algorithm.

\textit{Equivalence of hard M-step and K-Models update Step}: The M-step in the hard-EM algorithm takes the form $\hat\theta_j = \argmax_\theta \prod_{X\in C_j} f_{\theta}(X)$. Computing, we have
\begin{align*}
    \hat\theta_j &= \argmax_\theta \prod_{X\in C_j} f_{\theta}(X)\\
    &= \arg\max_\theta \sum_{X\in C_j}\log(f_{\theta_j}(X))\\
    &= \argmin_{\theta} |C_j|\cdot C - \sum_{X\in C_j}\log f_{\theta}(X)\\
    &=\argmin_\theta \sum_{X\in C_j} (C-\log f_\theta(X))\\
    &= \argmin_\theta \sum_{X\in C_j} M(\theta, X)\\
    &= \argmin_\theta M(\theta, C_j).
\end{align*}

So the set of parameters we get for the model fit in the hard-EM algorithm is the same as the set of parameters we get using the K-Models approach using the cluster-wide loss $M(\theta, C_j) = \sum_{X\in C_j} (C - \log f_\theta(X))$.

So, we have that the equivalence between the E-step and the K-Models assignment step, and the equivalence of the M-step and the K-Models update step which together give use the equivalence of hard-EM and K-Models for the given choice of loss function.
\end{proof}

Several mixture modeling approaches have been developed for time series clustering and classification tasks for which, to our knowledge, hard-clustering analogues have not been developed. Some examples include the work of Xiong and Yeung who develop an EM algorithm for clustering ARMA models \cite{armaEM}; the work of Ren and Barnett who develop an EM algorithm for a Wishart mixture model for clustering time series based on their autocovariance matrices \cite{ren2022autoregressive}; and the work of Coke and Tsao who develop an EM algorithm for a random effects model for clustering electricity consumers based on their consumption \cite{coke2010random}. Theorem~\ref{thm:hardem} says that each of these mixture models has a hard-clustering K-Models analogue. The corresponding K-Models algorithm can then potentially be used for initialization or as its own clustering algorithm. We show in the coming sections how this can be done for clustering ARMA models.



\section{K-AR(\textit{p}) Models}\label{sec:ar}
We start our time series clustering approach by discussing AR($p$) models, a special case of the more general ARMA($p,q$) model. Unlike general ARMA models and MA models, AR($p$) models do not have any unobserved covariates. This allows us to take a more direct approach to their estimation and more easily consider robust regression techniques for their estimation.

Recall that a time series $X_{t}$ is AR($p$) if it takes the form
\begin{equation*}
    X_t = a_t + \sum_{i=1}^p \phi_i X_{t-i},
\end{equation*}
where the $a_t$ form a white noise process. If we have several time series $X_{1,t},\ldots, X_{n,t}$ that follow the same AR($p$) process, then we can fit the parameters $\phi$ using the method of conditional least squares: 
\begin{equation*}
    \hat{\phi}_{LS} = \arg\min_{\phi}\sum_{j=1}^n\sum_{t=p+1}^T\left(X_{j,t} - \sum_{i=1}^p \phi_iX_{j,t-i}\right)^2,
\end{equation*}
and in fact, this approach is used by \cite{bulteel2016} for clustering VAR(1) models and by \cite{takano2020} for clustering VAR($p$) models. 

We extend this method and make it more robust to outlier time series by using least absolute deviations (LAD) instead of least squares regression: 
\begin{equation}
    \hat{\phi}_{LAD}= \arg\min_{\phi}\sum_{j=1}^n\sum_{t=p+1}^T\left\lvert X_{j,t} - \sum_{i=1}^p \phi_iX_{j,t-i}\right\rvert. \label{eq:lad}
\end{equation}

This can be written as a regression problem in cannonical form as 
\begin{equation*}
    \hat{\phi}_{LAD} = \arg\min_{\phi\in\RR^p} ||\mathbf{Y} - \mathbf{X}\phi||_1,
\end{equation*}
where $\mathbf{Y}' = (\mathbf{Y}_1,
\mathbf{Y}_2, \ldots, \mathbf{Y}_n)$ and $\mathbf{Y}_i' = (X_{i, p+1},\ldots, X_{i,T}),$
and 
\begin{equation*}
\quad \mathbf{X} = \begin{bmatrix}
\mathbf{X}_1\\
\mathbf{X}_2\\
\vdots \\
\mathbf{X}_n
\end{bmatrix}, \text{ where }
\mathbf{X}_i = \begin{bmatrix}
X_{i, p} & X_{i, p-1} & \cdots & X_{i, 1}\\
X_{i, p+1} & X_{i, p} & \cdots & X_{i, 2}\\
\vdots & \vdots & \ddots & \vdots\\
X_{i, T} & X_{i, T-1} & \cdots & X_{i, T-p}
\end{bmatrix}.
\end{equation*}

This reformulation of the problem allows us to use existing packages to solve the $L_1$ regression problem such as \texttt{L1pack} in \texttt{R}. If we use the sum of absolute deviations as our criteria in the assignment step and use Equation~(\ref{eq:lad}) in the update step of Algorithm~\ref{alg:kmodels}, then Theorem~\ref{thm:kmodels} tells us that Algorithm~\ref{alg:kmodels} will converge to a locally optimal solution. 

For a simulated example using 10 clusters of AR(2) time series and the LAD criteria, we were able to correctly recover all 10 clusters when we had 1000 time points for each time series. Lowering the number of time points to 100 for each time series reduces the recovery rate slightly. See Appendix~\ref{sec:restarts} for details.

\section{K-ARMA(\textit{p,q}) Models}\label{sec:arma}

Autoregressive Moving Average (ARMA) models are a well studied time series model. A time series $X_t$ is said to be ARMA($p,q$) if it has the following form:
\begin{equation*}
X_{t} = a_t + \sum_{i=1}^p \phi_i X_{t-i} + \sum_{j=1}^q \theta_j a_{t-j},
\end{equation*}
where the $a_t$ form a white noise process with constant variance $\sigma_a$.

In order to use the K-Models paradigm, however, we need to be able to fit a single ARMA($p,q$) model to a corpus of time series data. One standard method for fitting ARMA($p,q$) is the method of maximum likelihood under the additional assumption that $a_t\sim N(0, \sigma_a)$. We can extend the use of maximum likelihood for fitting a single ARMA model to a cluster of time series data $(X_{j,t})$ which amounts to minimizing the sums of squares criterion 
\begin{equation}
    (\hat\phi, \hat\theta) = \arg\min_{\phi,\theta}\sum_{j=1}^n\sum_{t=1}^T (X_{j,t} - \sum_{i=1}^p \phi_i X_{j,t-i} - \sum_{k=1}^q \theta_k a_{j,t-k})^2,\label{eq:cssq}
\end{equation} 
where $a_{j,t}$ is the white noise process associated with the time series $X_{j,t}$ (see Appendix~\ref{sec:mle} for a derivation of Equation~\ref{eq:cssq}).


We can turn this into a conditional sums of squares problem by conditioning on $X_{j,1}, \ldots, X_{j,p}$ for $j = 1, \ldots, n$, and assuming that the first $q$ innovations are $0$. That is $a_{j,1} = \cdots = a_{j,q} = 0$ for $j = 1,\ldots, n$. The conditional sums of squared criterion can be minimized using a generalized least-squares procedure, which we modify from the standard procedure\footnote{See Section 7.2.4 in \cite{box2016time} for a nice overview of the algorithm for a single time series.} in \cite{box2016time} for fitting ARMA($p,q$) models to get Algorithm~\ref{alg:karma}. Algorithm~\ref{alg:karma} allows us to fit ARMA($p,q$) models to a cluster of time series in the update step of Algorithm~\ref{alg:kmodels}, and Theorem~\ref{thm:kmodels} tells us that we should use the conditional sums of squares criteria for the assignment step to guarantee convergence of the K-Models algorithm for clustering ARMA($p,q$) time series.

\begin{algorithm}[t]
    \caption{Algorithm for fitting an ARMA($p,q$) model to a cluster of time series.}\label{alg:karma}
    \DontPrintSemicolon
    \KwData{Time series $X_{1,t}, \ldots, X_{n,t}$ of length $T$, AR order $p$ and MA order $q$}
    \KwResult{Parameters $\phi = (\phi_1, \ldots, \phi_p)$ and $\psi = (\psi_1, \ldots, \psi_q) = -\theta = (-\theta_1, \ldots, -\theta_q)$}
    \hrulefill\\
    \tcp{Initialization}
    Choose starting values $\phi^{(0)}$ and $\psi^{(0)}$\;
    Set $u_{j,t} = v_{j,t} = \epsilon_{j,t}=0$ for $t<\max(p,q)+1$\;
    $k \gets 0$\;
    \While{not converged}{
        $\epsilon_{j,t}^{(k)} \gets X_{j,t} - \sum_{i=1}^p\phi_{i}^{(k)}X_{j,t-i}  + \sum_{i=1}^q\psi_i \epsilon_{j,t-i}^{(k)}$\;
        $u_{j,t} \gets \sum_{i=1}^q \psi_i^{(k)}u_{j,t-i} + X_{j,t}$\;
        $v_{j,t} \gets \sum_{i=1}^q \psi_i^{(k)}v_{j,t-i}- \epsilon_{j,t}^{(k)}$\;
        ${\displaystyle(\Delta\phi, \Delta\psi) \gets \arg\min_{\Delta\phi, \Delta\psi}}\sum_{j=1}^n\sum_{t=\max{p,q}}^T\bigp{\epsilon_{j,t}^{(k)} - \sum_{i=1}^p \Delta\phi_i u_{j,t-i} - \sum_{i=1}^q \Delta\psi_i v_{j,t-i}}^2$ \;
        $\phi^{(k+1)} \gets \phi^{(k)}+\Delta\phi$\;
        $\psi^{(k+1)} \gets \psi^{(k)}+\Delta\psi$\;
        $k\gets k+1$\;
    }
\end{algorithm}

While our primary purpose in deriving Algorithm~\ref{alg:karma} is for its use in clustering, we point out that it can be used any time we want to fit a single ARMA(\textit{p,q}) model to a corpus of time series. In particular, this has been noted to be useful when studying sparse time series \cite{jha2015clustering,barbaglia2016commodity}.

We present Algorithm~\ref{alg:karma} for fitting an ARMA(\textit{p,q}) model where all of the time series are assumed to be the same length. This can be easily extended to time series of different lengths by computing all of the values of $\epsilon_{j,t}^{(k)}$, $u_{j,t}$ and $v_{j,t}$ up to time $T_i$, where $T_i$ is the length of time series $X_i$ and then computing the updates as
\begin{equation}
    {\displaystyle(\Delta\phi, \Delta\psi) \gets \arg\min_{\Delta\phi, \Delta\psi}}\sum_{j=1}^n\sum_{t=\max{p,q}}^{T_i}\bigp{\epsilon_{j,t}^{(k)} - \sum_{i=1}^p \Delta\phi_i u_{j,t-i} - \sum_{i=1}^q \Delta\psi_i v_{j,t-i}}^2.
\end{equation}
It is also possible to weight the time series during this computation by $1/T_i$ so that the longer time series don't have an outsize influence on the ARMA parameters for a cluster, or by another value of interest such as an estimate of the variance for time series $\sigma_{\epsilon_i}$.


The extension of the clustering algorithm from ARMA($p,q$) models to ARIMA($p,d,q$) models is straightforward when $d$ is known. During the model fitting process we start by taking $d$ differences of all of our time series $X_{1,t}, \ldots, X_{n,t}$ to get time series $Y_{1,t}, \ldots, Y_{n,t}$. We then run the K-ARMA($p,q$) algorithm on the resulting differenced time series $(Y_{i,t})$ to cluster time series. 

\subsection{An ARMA(1,1) Example}

As an example we consider 50 time series of length 200 generated from two ARMA(1,1) processes. The first 25 time series come from an ARMA(1,1) process with $\phi_1 = -0.4$ and $\theta = -0.2$. The rest of the time series come from an ARMA(1,1) process with $\phi_1 = 0.4$ and $\theta = 0.4$. Running the K-ARMAs algorithm on this dataset with $k=2$ and correctly specified model parameters $p=q=1$ we are able to completely recover the two clusters of models used in our data generation process (see Figure~\ref{fig:arma}).

\begin{figure}
    \centering
    \includegraphics[width=0.6\textwidth]{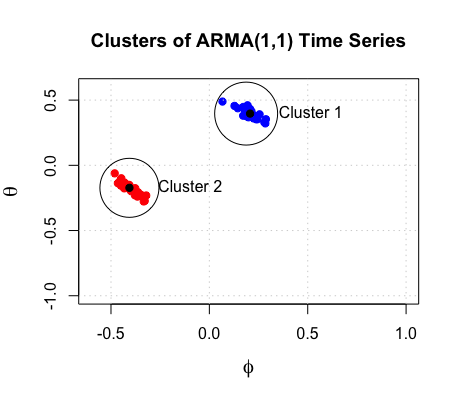}
    \caption{Each point represents the fitted ARMA(1,1) parameters colored by the clusters the K-Models algorithm identified.}
    \label{fig:arma}
\end{figure}

\section{Cluster Diagnostics}\label{sec:diagnostic}

Since the K-ARMA($p,q$) algorithm is unsupervised, it is useful to ask questions about the resulting clusters. The two questions that we focus on are:
\begin{enumerate}
    \item Is the model $M_i$ fitted to cluster $C_i$ appropriate for a time series $X\in C_i$?
    \item Is the model $M_i$ fitted to cluster $C_i$ appropriate for the entire cluster $C_i$?
\end{enumerate}

These questions are particularly important to answer if we want to use the models beyond just clustering the data such as for making forecasts for the time series in a cluster. We answer both of these questions using the Ljung-Box test \cite{ljungbox}. For the first question, we can apply the Ljung-Box test directly, where the residuals are calculated using the model $M_i$. For the second question, however, we need to extend the Ljung-Box test to work with models fit to clusters of data. We do this by establishing asymptotic distributions for statistics derived from the Ljung-Box statistic. The proofs of the results presented in this section extend the original proof of Box and Pierce \cite{boxpierce} and are detailed in Appendix~\ref{sec:proofs}.

Let $\hat{M}$ be an ARMA model fitted to a set of time series $C = \set{X_{1,t}, \ldots, X_{n,t}}$ each of length $T$. Denote the residuals of the time series $X_{i,t}$ under $\hat{M}$ as $\hat{a}_{i,t}$. We define the \textit{Grouped-Ljung-Box} statistic as:

\begin{equation}
    Q_{group}(\hat{r}) = T(T+2)\sum_{i=1}^n \sum_{l=1}^m (T-l)^{-1} \hat{r}_{i,l}^2, \label{eq:glb}
\end{equation}
where $m$ is the number of lags to consider, and $\hat{r}_{i,l}$ is the estimated $l$-lag auto-correlation of the residuals for time series $X_{i,t}$,
\begin{equation}
    \hat{r}_{i,l} = \bigp{\sum_{t=l+1}^T \hat{a}_{i,t} \hat{a}_{i,t-l}}/\bigp{\sum_{t=1}^T \hat{a}_{i,t}^2}.
\end{equation}

Note that (\ref{eq:glb}) is the sum of Ljung-Box statistics for the residuals of each individual time series in $C$, and in the case that a cluster $C$ has only one time series, $Q_{group}(\hat{r})$ is exactly equal to the Ljung-Box statistic. Using the prior of work of Box and Pierce, we can derive the asymptotic distribution of this statistic.

\begin{theorem}\label{thm:glb}
For $n$ independent time series $X_{1,t}, \ldots, X_{n,t}$ each of length $T$ with the fitted ARMA($p,q$) model $\hat{M}$, the Grouped-Ljung-Box statistic $Q_{group}(\hat{r})$ computed with $m$ lags asymptotically follows a $\chi^2_{(nm-p-q)}$ distribution if the model is correctly identified.
\end{theorem}

The proof of Theorem~\ref{thm:glb}, while quite similar to the derivation of Box and Pierce \cite{boxpierce}, is nontrivial because the $l$-lag autocorrelations in the residuals of different time series $\hat{r}_{i,l}$ and $\hat{r}_{j,l}$ are not independent. See Appendix~\ref{sec:proofs} for details.


Figure~\ref{fig:Qgroup} shows Monte-Carlo simulations of the $Q_{group}(\hat{r})$ statistic plotted with its corresponding asymptotic $\chi^2$ distribution. Visually, we see that as $T$ increases from $200$ to $2000$ the asymptotic approximation gets better, as we might expect for a limiting result. Additionally, a hypothesis test performed using the asymptotic result appears to be a conservative test when the value of $T$ is smaller.

\begin{figure}[b]
    \centering
    \makebox[\textwidth][c]{
    \includegraphics[width=0.5\textwidth]{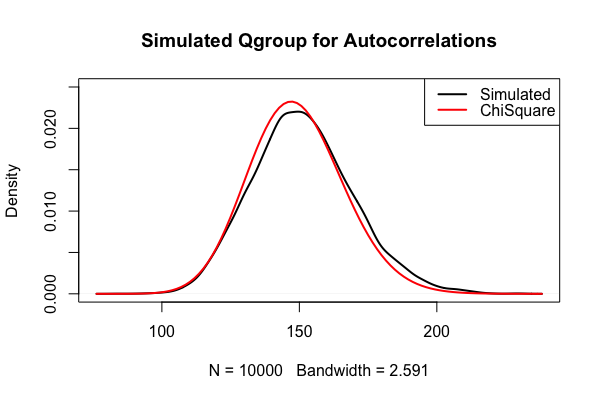}
    \includegraphics[width=0.5\textwidth]{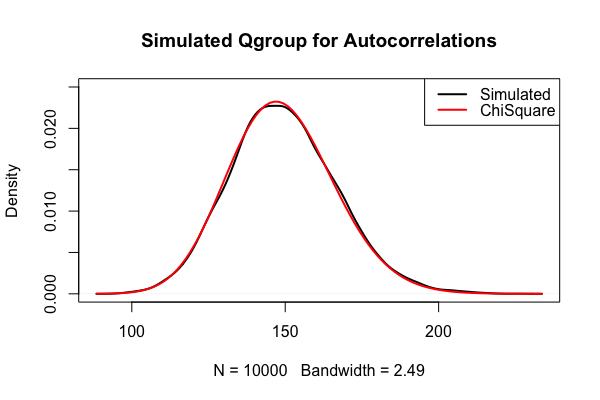}
    }
    \caption{Distribution of 10000 Monte-Carlo estimates of $Q_{group}(\hat{r})$ using $m=15$ lags for 10 AR(1) time series of length $T=200$ (left) and for 10 AR(1) time series of length $T=2000$ (right).}
    \label{fig:Qgroup}
\end{figure}

The asymptotic result presented in Theorem~\ref{thm:glb} suggests a portmanteau test for the goodness of fit of model $M_i$ to cluster $C_i$. That is, we can test the hypotheses
\begin{align*}
    H_0:\quad &\text{The are no autocorrelations in the residuals},\\
    H_1:\quad&\text{There are autocorrelations in the residuals}.
\end{align*}
by rejecting $H_0$ at the $\alpha$-level of significance if and only if $Q_{group} > \chi^2_{(nm-p-q)}(1-\alpha)$. Autocorrelations in the residuals would indicate that the model chosen for a cluster is not adequately explaining the behavior of the time series in that cluster. This may indicate model misspecification for a cluster, which in turn may indicate the grouping of multiple distinct clusters into a single cluster during the K-ARMA clustering process.




We can also generalize Theorem~\ref{thm:glb} to a goodness of fit test about all of the fitted models and their parameters. Let $C = \set{C_1, C_2, \ldots, C_k}$ be a set of clusters of sizes $n_1, n_2, \ldots, n_k$ respectively and fitted ARMA($p_i, q_i)$ models $\hat{M}_i$. Denote the time series in cluster $C_i = \set{X_{i,1,t}, X_{i,2,t}, \ldots X_{i,n_i, t}}$. Let $\hat{r}_{i,j,l}$ be the $l$-lag autocorrelation of the residuals for time series $X_{i,j,t}$ under its appropriate model $\hat{M}_i$. Then we can define the statistic $$Q_{total}(\hat{r}) = T(T+2)\sum_{i=1}^k\sum_{j=1}^{n_i}\sum_{l = 1}^m (T-l)^{-1}\hat{r}_{i,j,l}^2$$

\begin{corollary}\label{thm:Qtotal}
The statistic $Q_{total}(\hat{r})$ asymptotically follows a $\chi^2_{\nu}$ distribution where the degrees of freedom $\nu = \sum_{i=1}^k (n_im-p_i-q_i)$ if the models $\hat{M}_i$ are correctly identified.
\end{corollary}

Monti showed in \cite{monti1994proposal} that if we have a time series $X_{t}$ and a proposed ARMA($p,q$) model, then we can replace the fitted autocorrelations in the Ljung-Box statistic with the partial autocorrelations, denoted $\hat{\pi} = (\hat{\pi}_1, \ldots, \hat\pi_m)$ and the resulting statistic $$Q(\hat\pi) = T(T+2)\sum_{l=1}^m(T-l)^{-1} \hat\pi_{l}^2$$ asymptotically follows a $\chi^2_{m-p-q}$ distribution. We can use this result to extend the grouped version of the Ljung-Box statistic to use the partial autocorrelations as well.

If we let $\hat{\pi}_{i,k}$ denote the estimated $k$-lag auto-correlation of the residuals of time series $X_{i,t}$, then we can use the statistic $Q_{group}(\hat\pi)$ to perform a statistical test.

\begin{corollary}\label{thm:Qpacf}
For independent time series $X_{1,t}, \ldots, X_{n,t}$ each of length $T$ with the fitted ARMA($p,q$) model $\hat{M}$, the statistic $$Q_{group}(\hat{\pi}) = T(T+2)\sum_{i=1}^n \sum_{k=1}^m (T-k)^{-1}\hat\pi^2_{i,k}$$ asymptotically follows a $\chi^2_{(nm-p-q)}$ distribution if the model is correctly identified.
\end{corollary}

\begin{figure}[tb]
    \centering
    \makebox[\textwidth][c]{
    \includegraphics[width=0.5\textwidth]{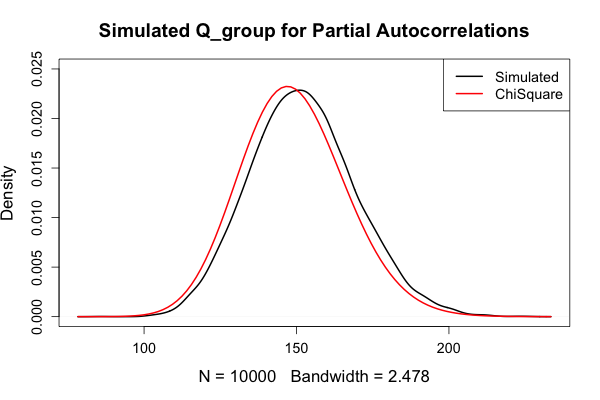}
    \includegraphics[width=0.5\textwidth]{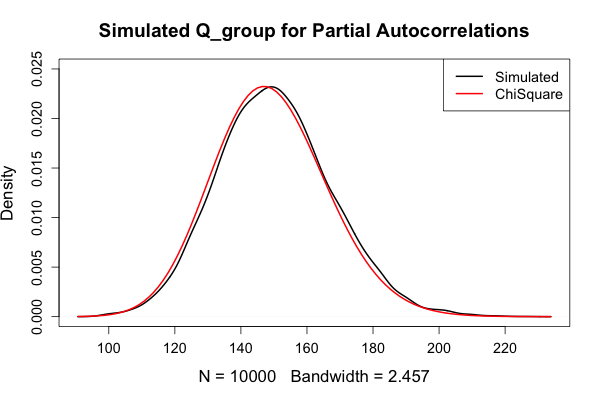}
    }
    \caption{Distribution of 10000 Monte-Carlo estimates of $Q_{group}(\hat{\pi})$ using $m=15$ lags for 10 AR(1) time series of length $T=200$ (left) and for 10 AR(1) time series of length $T=2000$ (right).}
    \label{fig:QgroupPacf}
\end{figure}

Figure~\ref{fig:QgroupPacf} shows Monte-Carlo simulations of the $Q_{group}(\hat\pi)$ statistic plotted with its corresponding asymptotic $\chi^2$ distribution. These plots appear very similar to the ones we constructed for $Q_{group}(\hat{r})$ in that as $T$ increases the approximation gets better, and the hypothesis test based on the asymptotic result appears to be a conservative test when the value of $T$ is smaller.

The extension also works when we have time series of different lengths and when we want to perform a goodness of fit test across multiple clusters of time series with different models yielding a partial correlation analogue to Corollary~\ref{thm:Qtotal}.


\begin{corollary}\label{thm:Qpacftotal}
Under the same conditions as Corollary~\ref{thm:Qtotal}, and denoting the estimated $l$-lag partial autocorrelation of time series $X_{i,j,t}$ as $\hat{\pi}_{i,j,t}$, the statistic $$Q_{total}(\hat{\pi}) = T(T+2)\sum_{i=1}^k \sum_{j=1}^{n_i}\sum_{l = 1}^m (T-l)^{-1}\hat\pi^2_{i,j,l}$$ is asymptotically $\chi^2_{\nu}$, where $\nu = \sum_{i=1}^l (n_im-p_i-q_i)$.
\end{corollary}

We remark that $Q_{total}(\hat{r})$ and $Q_{total}(\hat{\pi})$ can be used for simultaneous hypothesis testing.

\subsection{ARMA(1,1) Example with an Outlier}

We illustrate the use of $Q_{group}(\hat{r})$ for assessing the fit of a model $M_i$ to a cluster of time series $C_i$ with the ARMA(1,1) example from the previous section. We do this by running the K-ARMA(1,1) algorithm for two clusters with an additional time series that was generated from an ARMA(1,1) process with $\phi = 0.2$ and $\theta = -0.2$. When we plot the fitted parameters to each time series and color them by the fitted cluster, we see that added outlier is far from any other time series (See Figure~\ref{fig:outlier}).

While at a glance, the fitted cluster models are still able to separate the two clusters, the $Q_{group}(\hat{r})$ statistics tell us a different story about the fit of the models. For $20$ lags, the red cluster has $Q_{group}^{red}(\hat{r}) = 462.4082$, whereas $Q_{group}^{blue}(\hat{r}) = 750.0753.$ Under the Grouped-Ljung-Box test, these yield p-values of $p_{red} = 0.8717$ and $p_{blue}<10^{-11}.$
This indicates that the clustering model for the blue cluster is poorly specified, likely as a result of the outlier. 

Looking at the individual Ljung-Box tests for the blue cluster, we identify a time series with a statistic $Q_{outlier} = 239.61314$, far above the rest of the values, and indicating a poor fit for that time series. Identifying the outlier, removing it, and refitting our model on the remaining data in the cluster yields slightly shifted parameters and test statistic $Q_{group}^{blue'} = 496.7404$ with corresponding $p$-value $p_{blue'} = 0.5075057$ indicating a better model fit.

\begin{figure}[htb]
    \centering
    \includegraphics[width=0.6\textwidth]{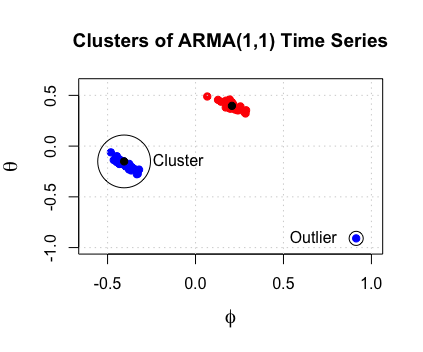}
    \caption{Each point represents the ARMA(1,1) parameters fitted to an individual time series. The points are colored according to the clusters assigned by the K-ARMA algorithm on data generated from 2 independent ARMA(1,1) processes with an outlier.}
    \label{fig:outlier}
\end{figure}

\subsection{Data Drift Application} 

In the machine learning literature there is a general interest in studying and identifying distributional drift (also called concept drift) \cite{lu2018learning}. Distributional drift occurs when the data that a model is trained on is different from subsequently observed data, usually a testing set \cite{conceptdrift}. Since the models developed using our clustering method can be used for classifying new time series we may want to keep an eye out for distributional drift which can occur in a few ways:
\begin{enumerate}
    \item If a new cluster appears.
    \item If new observations in a cluster's parameters change.
    \item If the proportion of time series in the different clusters changes.
\end{enumerate}

As the third point doesn't affect the clustering models, the Ljung-Box statistic is most useful in identifying the first two types of drift. If a new cluster appears, we would be able to identify this much in the same way we identified the outlier in our ARMA(1,1) example (which can be viewed as a new cluster). That is, we would find that the time series belonging to an unobserved cluster would have very large Ljung-Box statistics. If a clusters parameters drift with newly observed data, this would be more difficult to spot and would likely require many observations or significant drift before this becomes readily apparent in the Ljung-Box statistic unless the drift is happening quickly.

\section{Empty Clusters and Model Vanishing}\label{sec:vanishing}

With a K-Models approach there can be instances where we are unable to fit a model to a cluster of data. This can occur when we have more parameters than data in a cluster making us unable to fit a model to a cluster. As a subset of this, we have the special case when no data points are assigned to a cluster of data, which is called the empty cluster problem. In either case we are left unable to fit a model for the cluster assignment step. We call this phenomenon \textit{model vanishing}. This has been observed for $k$-means algorithms, and mitigation techniques have been developed including empty cluster reassignment \cite{gxkmeans, bock2008}.

In our K-ARMA clustering algorithm, we can fit an ARMA model to a single time series, so in practice we only run into model vanishing when no time series are assigned to a cluster in the assignment step. Empirically we have found that model vanishing occurs more frequently when the true number of clusters is less than the number of clusters learned by the K-ARMA algorithm. We also found that it occurs more frequently when we use random partitioning to initialize the clusters than when we use a prototype initialization method. Of all the factors considered in our experiments, the initialization method seemed to have the biggest impact on the model vanishing behavior. See Table~\ref{tab:simulation} for a summary of our simulations and Appendix~\ref{sec:restarts} for a more detailed description of the experiments.

Something we found particularly interesting is that when we selected a number of clusters larger than the true number of underlying generative processes, the partitioning initialization often reduced the number of non-empty clusters to the number of underlying generative processes. For example, when we clustered 4-AR(2) models with $k=10$ using the partitioning initialization and least squares parameter estimation, we got clusters in the range $3-6$, 40\% of which had 4 clusters and only 7 of which had 3 clusters. This seems to suggest that if we have a known number of underlying clusters then we should go with the prototype method of initialization and if we have an unknown number of underlying clusters then we should choose a large starting number of clusters and use the random partitioning method to get something close to the true number of clusters.

\begin{table}[htb]
  \caption{Average number of resulting clusters when 100 K-ARMA clusters are formed for various values of $k$. For Initialization, Pr = Prototype method, and Pa = Random Partitioning. For details see Appendix~\ref{sec:restarts}}
  \label{tab:simulation}
    \small%
	\centering%
  \begin{tabu}{%
	r%
	*{8}{c}%
	}
  \toprule
   Actual Model & \rotatebox{90}{Loss} & \rotatebox{90}{Init.} & \rotatebox{90}{k = 2} &   \rotatebox{90}{k = 3} &   \rotatebox{90}{k = 4} &   \rotatebox{90}{k = 5} &   \rotatebox{90}{k = 7} & \rotatebox{90}{k = 10}\\
  \midrule
  2-AR(2)     & L1 & Pr & 2.00 & 3.00 & 4.00 & 4.99 & 7.00 & 10.00 \\
  2-AR(2)     & L2 & Pr & 2.00 & 2.98 & 4.00 & 5.00 & 6.97& 9.70 \\
  2-AR(2)     & L1 & Pa & 2.00 & 2.38 & 2.61 & 2.71 & 2.74 & 3.33 \\
  2-AR(2)     & L2 & Pa & 2.00 & 2.29 & 2.55 & 2.62 & 2.73 & 3.10 \\
  4-AR(2)     & L1 & Pr & 2.00 & 3.00 & 3.99 & 4.97 & 6.97 & 9.96 \\
  4-AR(2)     & L2 & Pr & 2.00 & 3.00 & 3.97 & 4.98 & 6.97 & 9.95 \\
  4-AR(2)     & L1 & Pa & 2.00 & 3.00 & 3.34 & 3.73 & 4.07 & 4.63 \\
  4-AR(2)     & L2 & Pa & 2.00 & 2.98 & 3.33 & 3.77 & 4.08 & 4.46 \\
  10-AR(2)    & L1 & Pr & 2.00 & 3.00 & 4.00 & 5.00 & 7.00 & 9.98 \\
  10-AR(2)    & L2 & Pr & 2.00 & 3.00 & 4.00 & 5.00 & 6.99 & 9.98 \\
  10-AR(2)    & L1 & Pa & 2.00 & 3.00 & 4.00 & 5.00 & 6.79 & 8.19 \\
  10-AR(2)    & L2 & Pa & 2.00 & 3.00 & 4.00 & 5.00 & 6.69 & 8.04 \\
  2-ARMA(1,1) & L2 & Pr & 2.00 & 2.88 & 3.76 & 4.37 & 5.65 & 7.28 \\
  2-ARMA(1,1) & L2 & Pa & 2.00 & 2.04 & 2.05 & 2.08 & 2.24 & 2.90 \\
  4-ARMA(1,1) & L2 & Pr & 2.00 & 2.88 & 3.89 & 4.76 & 6.64 & 8.93 \\
  4-ARMA(1,1) & L2 & Pa & 1.52 & 2.36 & 3.07 & 3.64 & 4.17 & 4.59 \\
  \bottomrule
  \end{tabu}%
\end{table}

\section{Real Data Experiments}\label{sec:realdata}

We now show how our method can be used to cluster real world data. We do this using three open source datasets\footnote{Downloaded from \href{https://www.csee.umbc.edu/~kalpakis/TS-mining/ts-datasets.html}{https://www.csee.umbc.edu/~kalpakis/TS-mining/ts-datasets.html} on 10 April 2022.} compiled by Kalpakis et al. for their paper on clustering ARIMA time series using cepstral coefficients \cite{cepstral} and subsequently used by Xiong and Yeung in their paper exploring ARMA mixture modeling \cite{armaEM}. In order to compare the fitted clusters to expected clusters we use the same cluster similarity measure as was used by Kalpakis and Xiong which is defined as 
$$Sim(A, B) = \frac{1}{k}\sum_{i=1}^k \max_{1\leq j\leq k} Sim(A_i, B_j), \quad \text{ where }\quad Sim(A_i, B_j) = \frac{2\cdot|A_i \cap B_j|}{|A_i| + |B_j|},$$
where $A = \set{A_1, \ldots, A_k}$ and $B = \set{B_1, \ldots, B_k}$ are two clusterings of a set of data. The similarity measure $Sim(A,B)$ falls in the range $[0,1]$, with a $1$ indicating that the clusterings $A$ and $B$ are identical, and lower scores indicating weaker overlap. The similarity measure $Sim(A,B)$ is not symmetric so we will use the first input in $Sim(A,B)$ to always be the "known" or "expected" clusters and the second input to be the clusters fitted by our method. Because we are only guaranteed a local optimum by the K-ARMA algorithm, we will run it 10 times for each dataset with different random initializations and select a clustering that achieves the best score.

\subsection{Personal Income Data}

The Bureau of Economic Analysis maintains data on the personal income levels of each state in the USA \cite{BEA}. Looking at the  per capita income between the years 1929-1999, it has been suggested that the east coast states, California, and Illinois form a set of "high income growth" states and that the midwest states form a set of "low growth" states \cite{cepstral}. 

For the purposes of comparison, we took the same preprocessing approach as \cite{cepstral}, where the authors first take a rolling mean of the time series with a window size of 2 to smooth the data, followed by taking a log transform of the resulting data, and then using ARIMA(1,1,0) models to cluster the processed data. This left us with Fitted Clusters 1A and 2A in Table~\ref{tab:income}, which had a similarity score of $0.78$ when compared to the expected clusters.

Additionally, because taking the rolling average removes some of the signal in the time series, we clustered the log-transformation of the data without taking the rolling average and using a higher order model instead. We did this using ARIMA(5,1,0) models which yielded the Fitted Clusters 1B and 2B in Table~\ref{tab:income} which had a similarity score of $0.90$ when compared to the expected clusters. This is the same clustering result that was achieved by Xiong and Yeung \cite{armaEM}, and a better clustering result than the clustering of Kalpakis et al. \cite{cepstral}.

\begin{table}[hb]
  \caption{The expected clusters of states for the personal income data as suggested by \cite{cepstral, armaEM} and the fitted clusters given by the K-ARMA algorithm. The states that did not end up in the "expected" cluster are in red.}
  \label{tab:income}
  \small%
	\centering%
\makebox[\textwidth][c]{
  \begin{tabu}{%
	ll
	}
  \toprule
  & States\\
  \midrule
  Expected Cluster 1 & CT, DC, DE, FL, MA, ME, MD, NC, NJ, NY, PA, RI, VA, VT, WV, CA, IL\\
  Expected Cluster 2 & ID, IA, IN, KS, ND, NE, OK, SD\\
  \midrule
  Fitted Cluster 1A & CT, DC, DE, FL, MA, ME, MD, NC, NJ, NY, PA, RI, VA, VT, WV, CA, IL, \textcolor{red}{IN, KS, NE, OK}\\
  Fitted Cluster 2A & ID, IA, ND, SD\\
  \midrule
  Fitted Cluster 1B & CT, DC, DE, FL, MA, ME, MD, NC, NJ, NY, PA, RI, VA, VT, WV, CA, IL, \textcolor{red}{KS, OK}\\
  Fitted Cluster 2B & ID, IA, IN, ND, NE, SD\\
  \bottomrule
  \end{tabu}%
  }
\end{table}





\subsection{Temperature Dataset}

The temperature dataset that we consider was compiled by the National Oceanic and Atmospheric Administration (NOAA) and contains daily temperatures for the year 2000 at locations in Florida, Tennessee, and Cuba. In Kalpakis et al. they group these locations into two groups \cite{armaEM}. The first, contains locations for Northern Florida and Tennessee, and the second contained Southern Florida and Cuba. 

Using the model specifications of Kalpakis et al., we fit 2 clusters using ARIMA(2,1,0) models to the temperature data which resulted in clusters. Unlike Kalpakis et al., we decided not to preprocess the data using a moving window. The resulting two clusters are shown in Figure~\ref{fig:temperatures}, and had a similarity of $0.933$ with the expected clusters. The two locations that did not fall into their expected cluster were in northern Florida which is geographically on the border between the clusters. We also decided to fit 3 clusters to the temperature data, and when we did this all of locations in Cuba ended up in their own cluster.

\begin{figure}[hbt]
    \centering
    \makebox[\textwidth][c]{
    \includegraphics[height = 1.9in]{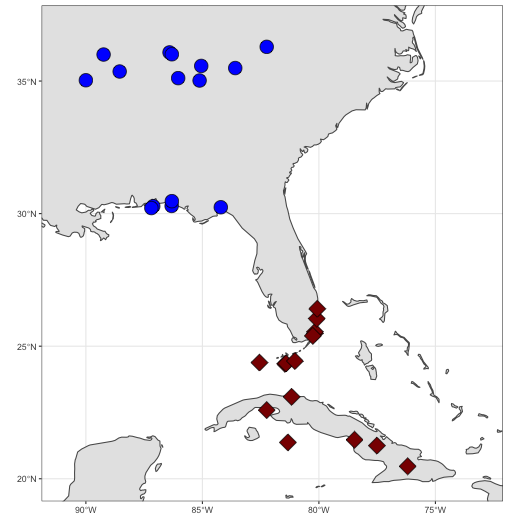}
    \includegraphics[height = 1.9in]{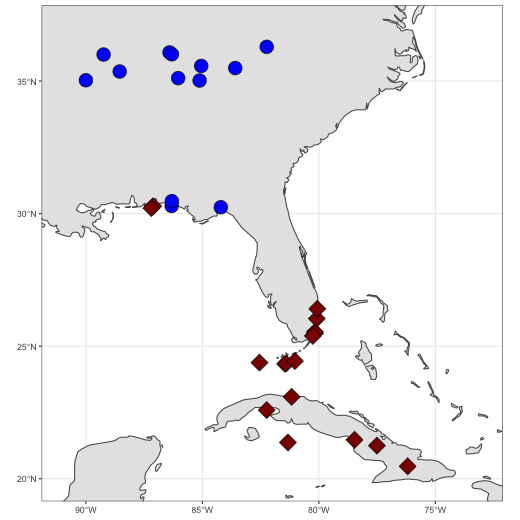}
    \includegraphics[height = 1.9in]{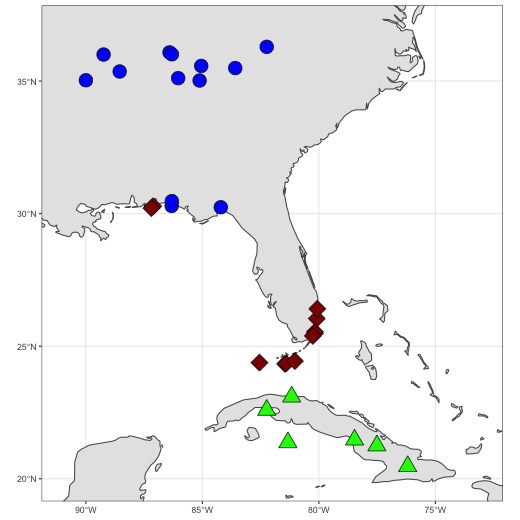}
    }
    \caption{Locations of the temperature collection sites colored by expected cluster (left), by fitted cluster using two clusters (center), and by fitted cluster using three different clusters (right).}
    \label{fig:temperatures}
\end{figure}


\subsection{Population Dataset}

The population of each states in the USA is well documented by the US Census Bureau and forms a time series. Between the years 1900-1999, some states saw their population grow rapidly, while other states saw their populations stagnate. In their paper, Kalpakis et al. make an argument that a subset of 20 states can be reasonably separated into exponential population growth states and states with stabilizing populations. These are listed as Expected Cluster 1 and Expected Cluster 2 in Table~\ref{tab:population}, respectively.

For comparison, we again used the same preprocessing as Kalpakis et al. and  took a similar approach as for the personal income dataset to cluster these states based on their population trend by first taking a rolling mean of the time series with a window size of 2 to smooth the data, followed by taking a log transform of the resulting data, and then using ARIMA(1,1,0) models to cluster the processed data. This left us with Fitted Clusters 1A and 2A in Table~\ref{tab:income}, which had a similarity score of $0.85$ when compared to the expected clusters which showed significantly better performance than either the ARMA mixture modeling approach or the cepstral coefficient clustering approach.

\begin{table}[htb]
  \caption{The expected clusters of states for the population data as suggested by \cite{armaEM, cepstral} and the fitted clusters given by the K-ARMA algorithm. The states that did not fall into their "expected" cluster are in red.}
  \label{tab:population}
  \small%
	\centering%
  \begin{tabu}{%
	ll
	}
  \toprule
  & States\\
  \midrule
  Expected Cluster 1 & CA, CO, FL, GA, MD, NC, SC, TN, TX, VA, WA\\
  Expected Cluster 2 & IL, MA, MI, NJ, NY, OK, PA, ND, SD\\
  \midrule
  Fitted Cluster 1 & CA, CO, FL, GA, MD, NC, SC, TX, VA, WA, \textcolor{red}{MI, NJ}\\ 
  Fitted Cluster 2 & IL, MA, NY, OK, PA, ND, SD, \textcolor{red}{TN}\\
  \bottomrule
  \end{tabu}%
\end{table}

\subsection{Comparison to Other Methods}

Because we use the same similarity score and the same data as that used by \cite{cepstral}, and presumably that used by \cite{armaEM}, we can compare the similarity scores of the clusters that we fit using the K-ARMA method to the clusters fit using a mixture modeling approach and the cepstral coefficient clustering method. This is done in Table~\ref{tab:comparison}. We see that the K-ARMA method performs comparably to the other methods on the personal income dataset and the temperature dataset and quite a bit better than the other two methods on the population dataset.

\begin{table}[htb]
  \caption{The similarity scores of the clusters fitted using the K-ARMA method, an ARMA mixture modeling EM algorithm, and using cepstral coefficients. The highest similarity scores are in bold.}
  \label{tab:comparison}
  \small%
	\centering%
  \begin{tabu}{lccc}
  \toprule
  & K-ARMA & ARMA Mixture Model \cite{armaEM} & Cepstral Coefficients \cite{cepstral}\\
  \midrule
  Personal Income & \textbf{0.90} & \textbf{0.90} & 0.84\\
  Temperature Dataset & 0.93 & \textbf{1.00} & 0.93\\
  Population Dataset & \textbf{0.85} & 0.64 & 0.74 \\
  \bottomrule
  \end{tabu}%
\end{table}

\section{Discussion}\label{sec:discussion}

We presented a new general framework for clustering time series data using a model-based generalization of the K-Means algorithm and showed how it can be used to cluster ARMA models. We use our method on three real world data sets and show that it is competitive with similar clustering methods. Additionally we show the importance of initialization method on empty clusters and model vanishing. There are several avenues we see for extending these results:

\begin{itemize}
    \item All of our experiments and algorithms were conducted using a fixed order for our ARMA models across all of the time series. It would be ideal to incorporate mixed architecture model-based clustering where we allow different clusters to have different order ARMA models. One difficulty with this, however, would be the selection of the model order during the clustering process. One method worth exploring for solving this may be using higher order models and a Lasso-like penalty to zero-out some of the coefficients to select the order.
    \item We only considered time series of the same length in our simulations and examples. It would be good to incorporate time series of different lengths and explore weighting based on the length of the time series. One way this can be done is by weighting the Conditional Sums of Squares criterion by the inverse length of the time series.
    \item We can also extend our results for a variety of different loss functions. For ARMA clustering we used a conditional sum of squares estimate derived from the likelihood function, however this has been shown to be biased for short time series \cite{chung1993small}. Additionally the conditional sums of squares estimate is not robust to outliers. Because of these issues, we may want to consider the unconditional MLE, and more robust fitting procedures such as least absolute deviations for the ARMA fitting process.
    \item Although the Grouped-Ljung-Box test statistic that we developed for diagnosing the models fit to the clusters gives us an indication of the goodness of fit for the clusters, it is not on its own comprehensive. Other measures of goodness of fit and model appropriateness should be explored for analyzing the clusters produced by our algorithm such as the Pe\~{n}a Rodr\'{i}guez Portmanteau test \cite{pena2002}.
    \item It would also be interesting to incorporate the Grouped-Ljung-Box statistic into the clustering process, perhaps by using it to dynamically add and remove clusters for the automatic selection of $k$, the number of clusters.
    \item We saw that the initialization method had a large impact on the resulting clusters. Using a random partition initialization appears to do some form of regularization with respect to the number of clusters but the mechanism that causes this behavior is not well understood and should be explored further.
\end{itemize}

\section*{Acknowledgements}

Derek Hoare and Martin Wells  were partially supported by NIH grant \# R01GM135926-03, and David Matteson gratefully acknowledges financial support from the National Science Foundation Awards 1934985, 1940124, 1940276, and 2114143.


\bibliographystyle{unsrt}
\bibliography{ms}

\appendix

\section{Cluster-wide ARMA(\textit{p,q}) MLE Criterion}\label{sec:mle}
If we assume that the innovations follow a normal distribution with a constant variance $\sigma_\epsilon$, then we can use maximum likelihood estimation for fitting ARMA parameters to a cluster of time series $X_{1,t},\ldots, X_{n,t}$.

For the ARMA($p,q$) model, we have $$X_{j,t} = \epsilon_{j,t} + \sum_{i=1}^p \phi_i X_{t-i} + \sum_{k=1}^q \theta_k \epsilon_{t-k} \qquad \text{where} \quad \epsilon_{j,t} \overset{i.i.d.}{\sim} N(0,\sigma_\epsilon^2).$$ Under this model we can write the likelihood function as 

\begin{equation*}\hspace{-33pt}L(\phi,\theta,\sigma_\epsilon) = \prod_{j=1}^n \prod_{t=1}^T \frac{1}{\sqrt{2\pi}\sigma_\epsilon} \exp\set{\frac{-(X_{j,t} - \sum_{i=1}^p \phi_i X_{t-i} - \sum_{k=1}^q \theta_k \epsilon_{t-k})^2}{2\sigma_\epsilon^2}}
\end{equation*}
which yields the log likelihood 
\begin{align*}
\ell(\phi,\theta,\sigma_\epsilon) =& -\frac{nT}{2}\log(2\pi) - \frac{nT}{2}\log(\sigma_\epsilon^2) \\
&\quad - \frac{1}{2\sigma_\epsilon^2}\sum_{j=1}^n \sum_{t=1}^T (X_{j,t} - \sum_{i=1}^p \phi_i X_{t-i} - \sum_{k=1}^q \theta_k \epsilon_{t-k})^2.
\end{align*}

Maximizing the log-likelihood function then amounts to finding the parameters $\phi$ and $\theta$ that minimize the sums of squares criterion $$SSE = \sum_{j=1}^n \sum_{t=1}^T (X_{j,t} - \sum_{i=1}^p \phi_i X_{t-i} - \sum_{k=1}^q \theta_k \epsilon_{t-k})^2.$$ We can turn this into a conditional sums of squares problem by conditioning on $X_1, \ldots, X_p$, and assuming that the first $q$ innovations are $0$. This leaves us with the conditional sums of squares criterion 
\begin{equation*}
cSSE = \sum_{j=1}^n\sum_{t = max(p,q)+1}^T (X_{j,t} - \sum_{i=1}^p\phi_i X_{j,t-i} - \sum_{i=1}^q \theta_i \epsilon_{j,t-i})^2 
\end{equation*}
which is what we minimize in Algorithm~\ref{alg:karma}.

\section{Additional Proofs}\label{sec:proofs}

\textbf{Theorem~\ref{thm:glb}}. For $n$ independent time series $X_{1,t}, \ldots, X_{n,t}$ each of length $T$ with the fitted ARMA($p,q$) model $\hat{M}$, the Grouped-Ljung-Box statistic $Q_{group}(\hat{r})$ computed with $m$ lags asymptotically follows a $\chi^2_{(nm-p-q)}$ distribution if the model is correctly identified.

\begin{proof}
The proof generally follows the proof of the asymptotic distribution of the Ljung-Box Statistic in \cite{ljungbox}. We start by considering the true white noise process and the true correlations. If the $a_{i,t}$ are iid Normal with mean $0$, then according to \cite{anderson1942distribution, anderson1964asymptotic}, the limiting distribution of the autocorrelations in the white noise process for $X_{i,t}$, which we denote $r_{i}^{(m)} = (r_{i,1}, \ldots, r_{i,m})',$ is a normal distribution with mean $0$ and diagonal covariance matrix where $$Var(r_{i,l}) = \frac{T-l}{T(T+2)}.$$ Denote this covariance matrix of $r_{i}^{(m)}$ as $\Sigma^{(m)}$. Our assumptions also imply that the vectors $r_i^{(m)}$ and $r_j^{(m)}$ are independent for $i\neq j$. Letting $r_{long} = (r_1^{(m)'}, r_2^{(m)'}, \ldots, r_n^{(m)'})'$ be the stacked vector of the autocorrelations, we have $r_{long}\sim N(\textbf{0}_{nm}, I_n\otimes \Sigma^{(m)}).$

Box and Pierce showed in \cite{boxpierce} that there exists a matrix $A$ such that $\hat{r}_i^{(m)} \overset{d}{=} (I_m - A)r_i^{(m)} + O_p(1/T)$, where $A$ is an idempotent matrix of rank $p+q$. In particular, \cite{McLeod1983} shows in detail that the matrix $A$ takes the form $A = X(X'X)^{-1}X'$, where $X$ is the $m\times (p+q)$ matrix

\begin{equation}
    X = \begin{bmatrix}
    1 & 0 & \cdots & 0 &\rvline& 1 & 0 & \cdots & 0\\
    \phi'_1 & 1 & \cdots & 0 &\rvline & \theta'_1 & 1 & \cdots & 0\\
    \phi'_2 & \phi'_1 & \cdots & 0 &\rvline & \theta'_2 & \theta'_1 & \cdots & 0\\
    \vdots & \vdots & \ddots & \vdots & \rvline & \vdots & \vdots & \ddots & \vdots\\
    \phi'_{m-1} & \phi'_{m-2} & \cdots & \phi'_{m-p} & \rvline & \theta'_{m-1} & \theta'_{m-2} & \cdots & \theta'_{m-q}
    \end{bmatrix},
\end{equation}
with the $\phi'$ and $\theta'$ are constructed from the ARMA(p,q) polynomials $$\phi(B)=1-\phi_1B-\cdots-\phi_pB^p\quad \text{and}\quad \theta(B)=1-\theta_1B-\cdots -\theta_qB^q$$ so that $$1/\phi(B) = \sum_{i=1}^n \phi'_i B^i \quad\text{and} \quad1/\theta(B) = \sum_{i=1}^n \theta'_i B^{i}.$$ 

We will construct an idempotent matrix $C$ of rank $p+q$ so that $$\hat{r}_{long} \overset{d}{=} (I_{nm} - C)r_{long} + O_p(1/T)$$ which will yield the desired asymptotic result. From Box and Pierce \cite{boxpierce} we have that $$\hat{r}_{i}^{(m)} \overset{d}{=} r_{i}^{(m)} + X(\beta - \hat{\beta}) + O_p(1/T) \qquad 1\leq i \leq n,$$ where $\beta = (\phi_1, \ldots, \phi_p, \theta_1, \ldots, \theta_q)'$, and $\hat\beta = (\hat\phi_1, \ldots, \hat\phi_p, \hat\theta_1, \ldots, \hat\theta_q)'$ is our set of fitted parameters. Consequently, we have that 
\begin{equation}\label{eq:rlong}
\hat{r}_{long} \overset{d}{=} r_{long} + \mathbf{1}_n \otimes X(\beta - \hat{\beta}) + O_p(1/T).
\end{equation}
Consider the matrix $D = \mathbf{1}_n\otimes X$, and let $C = D(D'D)^{-1}D'$. Computing, we have that $C = \frac{1}{n}(\mathbf{1}_n\mathbf{1}_n^T)\otimes X(X'X)^{-1}X'.$

We now left multiply both sides of Equation (\ref{eq:rlong}) by $C$, performing a similar computation to Box and Pierce \cite{boxpierce}. When we do this, we get that $$C\hat{r}_{long} \overset{d}{=} 0 + O_p(1/T).$$ Additionally we have that 
\begin{align*}
    C(r_{long} + \one{n}\otimes X(\beta - \hat\beta) & = Cr_{long} + C\one{n}\otimes X(\beta - \hat\beta)\\
    &= Cr_{long} + \bigp{\frac{1}{n}(\mathbf{1}_n\mathbf{1}_n^T)\otimes X(X'X)^{-1}X'}\bigp{\one{n}\otimes X(\beta - \hat\beta)}\\
    &= Cr_{long} + \bigp{\frac{1}{n}(\mathbf{1}_n\mathbf{1}_n^T\one{n})}\otimes \bigp{X(X'X)^{-1}X'X(\beta - \hat\beta)}\\
    &= Cr_{long} + \one{n}\otimes X(\beta - \hat\beta)\\
    &\overset{d}{=} Cr_{long} + \hat{r}_{long} - r_{long} + O_p(1/T).
\end{align*}

Together, we get 
\begin{align*}
    0 &\overset{d}{=} Cr_{long} + \hat{r}_{long} - r_{long} + O_p(1/T)\\*
    \implies \qquad \hat{r}_{long} &\overset{d}{=} (I_{nm} - C)r_{long} + O_p(1/T).
\end{align*}
Consequently, $\hat{r}_{long}$ is asymptotically $N(\mathbf{0}_{nm},(I_{nm}-C)(I_n\otimes \Sigma^{(m)}))$. Call the covariance matrix of the asymptotic distribution $\Gamma = (I_{nm}-C)(I_n\otimes \Sigma^{(m)}).$ We now observe that $$Q_{group}(\hat{r}) = T(T+2)\sum_{i=1}^n \sum_{l=1}^m (T-l)^{-1}\hat{r}_{i,l}^2  =  \hat{r}_{long}' (I_n\otimes \Sigma^{(m)})^{-1} \hat{r}_{long}$$
is a quadratic form. Denoting $M = (I_n\otimes \Sigma^{(m)})^{-1},$ a little algebra yields that $\Gamma M \Gamma M \Gamma = \Gamma M \Gamma,$ a result of which is that $\hat{r}_{long}'M\hat{r}_{long} \sim \chi^2(\text{df} = \tr(\Gamma M))$. Some more algebra yields that $\tr(\Gamma M) = \rank(I_{nm} - C) = nm-p-q$. Thus, $Q_{group}(\hat{r})$ is asymptotically $\chi^2_{nm-p-q}.$
\end{proof}



\textbf{Corollary \ref{thm:Qtotal}.} The statistic $Q_{total}(\hat{r})$ asymptotically follows a $\chi^2_{\nu}$ distribution where the degrees of freedom $\nu = \sum_{i=1}^k (n_im-p_i-q_i)$ if the models $\hat{M}_i$ are correctly identified.

\begin{proof}
Let $Q_i = T(T+2)\sum_{j=1}^{n_i}\sum_{l = 1}^m (T-l)^{-1}\hat{r}_{i,j,l}^2$ be the Grouped-Ljung-Box statistic for the time series in cluster $C_i$ under the ARMA($p_i, q_i$) model $\hat{M}_i$. Then by Theorem~\ref{thm:glb} we have that $Q_i$ asymptotically follows a $\chi^2_{(n_im - p_i - q_i)}$ distribution. Furthermore, we have that the $Q_i$ are all mutually independent since the time series $X_{i,j,t}$ are independent and the fitted model $\hat{M}_i$ only depends on the time series in $C_i$. So $$Q_{total}(\hat{r}) = T(T+2)\sum_{i=1}^k\sum_{j=1}^{n_i}\sum_{l = 1}^m (T-l)^{-1}\hat{r}_{i,j,l}^2 = \sum_{i=1}^k Q_i$$ is asymptotically a sum of independent $\chi^2_{(n_im - p_i - q_i)}$ distributions and is hence asymptotically $\chi^2_{\nu}$ where $\nu = \sum_{i=1}^k (n_im - p_i - q_i)$.
\end{proof}

\textbf{Corollary \ref{thm:Qpacf}.} For independent time series $X_{1,t}, \ldots, X_{n,t}$ each of length $T$ with the fitted ARMA($p,q$) model $\hat{M}$, the statistic $$Q_{group}(\hat{\pi}) = T(T+2)\sum_{i=1}^n \sum_{k=1}^m (T-k)^{-1}\hat\pi^2_{i,k}$$ computed using $m$ lags asymptotically follows a $\chi^2_{(nm-p-q)}$ distribution if the model is correctly identified.

\begin{proof}
The proof of this result is very similar to the proof of Theorem~\ref{thm:glb}.
Let $r_i$, $\hat{r}_i$, $r_{long}$, $\hat{r}_{long}, X, A, C, \Sigma^{(m)}$ be as in the proof of Theorem~\ref{thm:glb}.
Let $\hat\pi_i = (\hat\pi_{i,1}, \ldots, \hat\pi_{i,m})'$ denote the fitted partial autocorrelations in the white noise process for the time series $X_{i,t}$. In \cite{monti1994proposal}, Monti shows that $\hat{\pi}^{(m)}_i \overset{d}{=} \hat{r}^{(m)}_i + O_p(1/T),$ and, as a consequence of this, $\hat{\pi}_i = (I_m - A)r_i + O_p(1/T)$. This result of Monti together with our construction in the proof of Theorem~\ref{thm:glb} implies that $\hat\pi_{long} = (\hat\pi_1', \ldots, \hat\pi_n')'$ satisfies 
\begin{align*}
    \hat\pi_{long} &\overset{d}{=} \hat{r}_{long} + O_p(1/T) \\
    &\overset{d}{=}(I_{nm} - C)r_{long} + O_p(1/T)\\ &\overset{d}{=} N\left(\mathbf{0}_{nm}, (I_{nm} - C)(I_n\otimes \Sigma^{(m)})\right) + O_p(1/T).
\end{align*}

That is, $\hat\pi_{long}$ is asymptotically $N\left(\mathbf{0}_{nm}, (I_{nm} - C)(I_n\otimes \Sigma^{(m)})\right)$. Observing that $Q_{group}(\hat{\pi}) = \hat{\pi}_{long}'(I_n\otimes\Sigma^{(m)})^{-1}\hat{\pi}_{long},$ it follows that $Q_{group}(\hat\pi)$ is asymptotically $\chi^2_{nm-p-q}$.
\end{proof}



\textbf{Corollary \ref{thm:Qpacftotal}.} Under the same conditions as Corollary~\ref{thm:Qtotal}, and denoting the estimated $l$-lag partial autocorrelation of time series $X_{i,j,t}$ as $\hat{\pi}_{i,j,t}$, the statistic $$Q_{total}(\hat{\pi}) = T(T+2)\sum_{i=1}^k \sum_{j=1}^{n_i}\sum_{l = 1}^m (T-l)^{-1}\hat\pi^2_{i,j,l}$$ is asymptotically $\chi^2_{\nu}$, where $\nu = \sum_{i=1}^l (n_im-p_i-q_i)$

\begin{proof}
This follows from Corollary~\ref{thm:Qpacf} and the argument we used to prove Corollary~\ref{thm:Qtotal}.
\end{proof}

\section{K-ARMA Simulation Experiments}\label{sec:restarts}
In this appendix we discuss in depth some experiments we conducted to better understand model-vanishing, cluster identification, and the need for random restarts. All of the time series used in our clustering experiments were generated using the \texttt{arima.sim} function in \texttt{R}.

\subsection{AR(\textit{p}) Models}

We start by considering clusters of AR($p$) models. In this section we describe how the AR(2) models were generated for the experiments conducted in Table~\ref{tab:simulation}.

\subsubsection{2-AR(2) Models}

We first consider clusters of $2-AR(2)$ models generated from the parameters $(\phi_1 = 0.7, \phi_2 = .25) $ and $(\phi_1 = -0.3, \phi_2 = 0.2 )$. We generated 25 models from each of these parameter sets to get 2 clusters of time series. Each time series was a fixed length of 100 time units.

\subsubsection{4-AR(2) Models}

We next consider $4-AR(2)$ models generated from the following parameters: $(\phi_1 = 0.7,\phi_2 = .2)$, $(\phi_1 = -0.3,\phi_2 = .2)$, $(\phi_1 = 0.4, \phi_2 = -.2)$, and $(\phi_1 = -0.2, \phi_2 = -.5)$.

\subsubsection{10-AR(2) Models}\label{sec:10AR}

Ten clusters of AR(2) models were constructed from the parameters in Table~\ref{tab:10AR2}. We generated 25 time series for each cluster with 1000 time points each.

\begin{table}[H]
  \caption{Parameters for 10 AR(2) models.}
  \label{tab:10AR2}
  \scriptsize%
	\centering%
  \begin{tabular}{%
	m{10pt}*{10}{m{23pt}<{\centering}}%
	}
  \hline
   \rotatebox{90}{Parameter} & \rotatebox{90}{Cluster 1} & \rotatebox{90}{Cluster 2}& \rotatebox{90}{Cluster 3}& \rotatebox{90}{Cluster 4}& \rotatebox{90}{Cluster 5}& \rotatebox{90}{Cluster 6}& \rotatebox{90}{Cluster 7}& \rotatebox{90}{Cluster 8}& \rotatebox{90}{Cluster 9}& \rotatebox{90}{Cluster 10}\\\hline
  $\phi_1$ &$-0.097$ & $-0.215$ & 0.419 & $-0.237$ & 0.273 & 0.403 & 0.281 & 0.144 & 0.105 & 0.861 \\
  $\phi_2$ & $-0.945$ & $-0.463$ & 0.206 & 0.135 & 0.640 & $-0.497$ & 0.500 & 0.824 & $-0.550$ & $-0.520$ \\\hline
  \end{tabular}%
\end{table}

\begin{figure}[H]
    \centering
    \includegraphics[width = 0.5\textwidth]{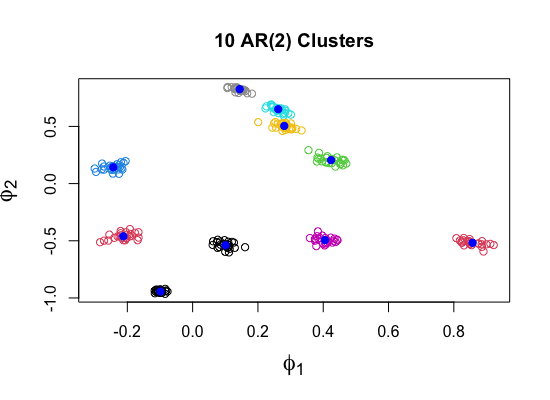}
    \caption{10 correctly identified clusters of AR(2) time series. The fitted parameters are in blue.}
    \label{fig:10AR2}
\end{figure}

\begin{figure}[H]
    \centering
    \includegraphics[width = 0.45\textwidth]{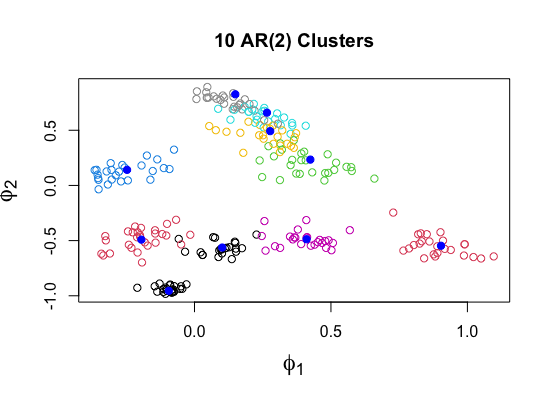}
    \includegraphics[width = 0.45
    \textwidth]{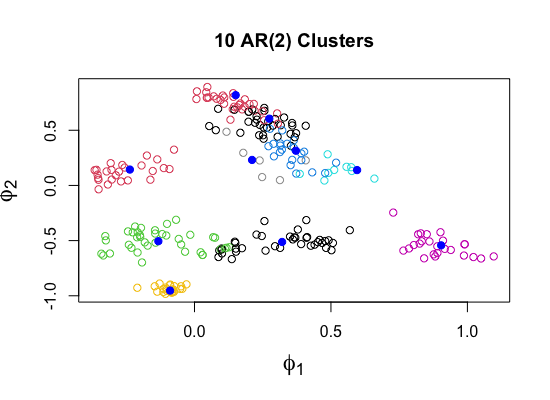}
    \caption{Two clusterings of the time series with only 100 time points for each series. The left clusters have $Sim(A,B) = 0.9$, whereas on the right the algorithm converged to a local solution and $Sim(A,B) = 0.82$}
    \label{fig:10AR100}
\end{figure}


If we look at fewer time points for each of these time series, the separation is less clear. This makes it more difficult to cluster the time series which can be seen in Figures~\ref{fig:10AR100}. Figure~\ref{fig:10AR100} also shows the importance of cluster initialization.









\subsection{ARMA(1,1) Models}

We also considered ARMA(1,1) models in this paper which we also generated using \texttt{arima.sim} in \texttt{R}.

\subsubsection{2-ARMA(1,1)}
For the clusters generated using 2 clusters of ARMA(1,1) models we used the parameters $(\phi_1 = -0.4, \theta_1 = 0.2)$ and $(\phi_1 = -0.2, \theta_1 = 0.4)$. We generated 25 time series from each of these parameter sets where each time series had 1000 time points.

\subsubsection{4-ARMA(1,1)}
For the clusters generated using 4 clusters of ARMA(1,1) models we used the parameters $(\phi_1 = -0.4, \theta_1 = 0.2)$, $(\phi_1 = -0.2, \theta_1 = 0.4)$, $(\phi_1 = 0.2, \theta_1 = 0.4)$, and $(\phi_1 = -0.2, \theta_1 = -0.4)$. We generated 25 time series from each of these parameter sets where each time series had 1000 time points.


\end{document}